\documentclass[12pt]{article}

\usepackage{graphicx, amssymb, latexsym, amsfonts, amsmath, lscape, amscd, pgfplots, capt-of, floatrow, stmaryrd,
amsthm, color, epsfig, mathrsfs, tikz, enumerate}

\newtheorem{theorem}{Theorem}[section]
\newtheorem{conjecture}[theorem]{Conjecture}
\newtheorem{corollary}[theorem]{Corollary}

\newtheorem{lemma}[theorem]{Lemma}
\newtheorem{proposition}[theorem]{Proposition}

\newtheorem{defn}[theorem]{Definition}
\newtheorem{observation}[theorem]{Observation}

\theoremstyle{definition}
\newtheorem{rem}[theorem]{Remark}
\newcommand\DELETE[1]{}

\pgfplotsset{compat = newest}

\begin{document}

%\begin{frontmatter}

\title{{\bf Counting the minimum number of arcs in an oriented graph having weak diameter $2$}}
\author{
{\sc Sandip Das}$\,^{a}$, {\sc Koushik Kumar Dey}$\,^{a}$, {\sc Pavan P D}$\,^{b}$, {\sc Sagnik Sen}$\,^{b}$\\
\mbox{}\\
{\small $(a)$ Indian Statistical Institute, Kolkata, India}\\
{\small $(b)$ Indian Institute of Technology Dharwad, India}
}

\date{}

\maketitle

\begin{abstract}
An oriented graph has weak diameter at most $d$ if every non-adjacent pair of vertices are connected by a directed $d$-path. The function $f_d(n)$ denotes the minimum number of arcs in an oriented graph on $n$ vertices having weak diameter $d$. Finding the exact value of $f_d(n)$ is a challenging problem even for $d = 2$. This function was introduced by Katona and Szeme{\'r}edi (1967), and after that several attempts were made to find its exact value by Znam (1970), Dawes and Meijer (1987), F{\"u}redi, Horak, Pareek and Zhu (1998), and Kostochka, Luczak, Simonyi and Sopena (1999) through improving its best known bounds. In that process, it was proved that this function is asymptotically equal to $n\log_2 n$ and hence, is an asymptotically increasing function. However, the exact value and behaviour of this function was not known. 

In this article, we observe that the oriented graphs with weak diameter at most $2$ are precisely the absolute oriented cliques, that is, analogues of cliques for oriented graphs in the context of oriented coloring. Through studying arc-minimal absolute oriented cliques we prove that $f_2(n)$ is a strictly increasing function. Furthermore, we improve the best known upper bound of $f_2(n)$ and conjecture that our upper bound is tight. This improvement of the upper bound improves known bounds involving the oriented achromatic number.
\end{abstract}


%use color instead of colour
\section{Introduction}
\label{sec:sec1}
In relation to the famous degree-diameter problem, Erd\H{o}s, R\'{e}nyi and S\'{o}s~\cite{erdos1966problem} defined the function $$h_d(n, k)=\min\{|E(G)| : \text{diam}(G) \leq d, \Delta(G) \leq k \text{ and }|V(G)|=n\}$$ as the minimum number of edges among graphs having diameter at most $d$ and maximum degree at most $k$. Its directed (oriented) analogue turns out to be very interesting even 
for small values of $d$, without even restricting the maximum degree. However, we need to recall a few notions before going into that.

\subsection{The function $f_d(n)$}
An \textit{oriented graph} $\overrightarrow{G}$ is a directed graph without any directed cycle of size $1$ or $2$. The sets of the vertices and arcs of $\overrightarrow{G}$ are denoted by $V(\overrightarrow{G})$ and $A(\overrightarrow{G})$, respectively.  Moreover, the \textit{weak diameter} of $\overrightarrow{G}$,  denoted by $\text{diam}_w(\overrightarrow{G})$, is the minimum $d$ such that given any two distinct vertices $u$ and $v$ of $\overrightarrow{G}$, there exists a directed path of length at most $d$ from $u$ to $v$ or from $v$ to $u$. Therefore, an analogue (ignoring the restriction on the maximum degree)  of $h_d$, in this case, is given by $$f_d(n)=\min\{|A(\overrightarrow{G})| : \text{diam}_w(\overrightarrow{G}) \leq d \text{ and }|V(\overrightarrow{G})|=n\}.$$

Finding the exact value of $f_d(n)$, even when $d=2$, turns out to be a challenging problem. On the other hand, the oriented graphs having weak diameter at most $2$ are objects of interest in the later discovered theory of oriented homomorphisms and colorings~\cite{sopena1997chromatic}.

\subsection{Relation with oriented coloring}
The notions of oriented homomorphisms, colorings, and chromatic number were introduced by Courcelle~\cite{courcelle1994monadic} in the series of papers that led to the illustrious Courcelle's theorem. A vertex mapping $\phi: V(\overrightarrow{G}) \to V(\overrightarrow{H})$ is a \textit{homomorphism} of $\overrightarrow{G}$ to $\overrightarrow{H}$ if for every arc $\overrightarrow{uv}$ of $\overrightarrow{G}$, $\overrightarrow{\phi(u)\phi(v)}$ is also an arc of $\overrightarrow{H}$. The \textit{oriented chromatic number} of $\overrightarrow{G}$, denoted by $\chi_o(\overrightarrow{G})$, is the minimum $|V(\overrightarrow{H})|$ such that $\overrightarrow{G}$ admits a homomorphism to $\overrightarrow{H}$. Observe that, an analogous definition of graph homomomorphism for undirected graphs yields a similar equivalent definition of the chromatic number~\cite{hell2004graphs}. Thus, oriented coloring is indeed a true generalization of the ordinary coloring. 

In a quest to find an analogue of a clique for oriented graphs with respect to oriented colorings, Klostermeyer and MacGillivray~\cite{klostermeyer2004analogues} 
defined an \textit{absolute oriented clique} $\overrightarrow{C}$ as an oriented graph satisfying $\chi_o(\overrightarrow{C}) = |V(\overrightarrow{C})|$. They characterized the absolute oriented cliques as follows. 
\begin{theorem}[Klostermeyer and MacGillivray~\cite{klostermeyer2004analogues}]
\label{thm:occhar}
    An oriented graph $\overrightarrow{C}$ is an absolute oriented clique if and only if $\overrightarrow{C}$ is an oriented graph having weak diameter at most $2$.
\end{theorem}
Hence, the objects of interest in studying the function $f_2(n)$ are absolute oriented cliques.

Incidentally, it turns out that the absolute oriented cliques are not so rare, as pointed out by Bensmail, Duffy and Sen~\cite{bensmail2017analogues}, who showed that almost all oriented graphs are absolute oriented cliques. Yet they are difficult to detect as Kirgizov, Duvignau and Bensmail~\cite{kirgizov2016complexity} showed that it is NP-hard to determine whether a given simple graph can be oriented as an absolute oriented clique or not. One notable structural result in this topic, due to Nandi, Sen and Sopena~\cite{nandy2016outerplanar}, settled a conjecture by Klostermeyer and MacGillivray~\cite{klostermeyer2004analogues} by showing that a planar absolute oriented clique can have at most $15$ vertices.

\begin{rem}
    The objects known as absolute oriented cliques, and their generalization - relative oriented cliques, appear in several contexts in the study of homomorphisms and colorings of oriented graphs. Apart from the few dedicated studies of oriented cliques~\cite{das2015oriented,das2018study,klostermeyer2004analogues,nandy2016outerplanar}, they are also frequently used to prove lower bounds of oriented chromatic number of graph families~\cite{duffy2019oriented,dybizbanski2020oriented,sopena1997chromatic}. The study of deeply critical oriented graphs uses absolute oriented cliques as the key constructions of the main results~\cite{BORODIN2001150,duffy2023}. Furthermore, the study of oriented complete coloring and achromatic number~\cite{PD202448,sopena2014complete} heavily relies on the understanding of the structures of absolute oriented cliques. In fact, the knowledge of the function $f_2(n)$ can greatly impact its research. 
\end{rem}

\subsection{Bounds of $f_2(n)$}

The problem of determining the function $f_2(n)$ was originally posed by Erd\H{o}s, R\'{e}nyi and S\'{o}s~\cite{erdos1966problem} in 1966 and later by Znam~\cite{znam1970minimal} and Dawes and Meijer~\cite{dawes1987arc}. For undirected graphs, the answer to the analogous question, that is, determining the exact value of $h_2(n, n - 1)$, is trivial: a graph with diameter at most 2 has at least $n - 1$ edges and the star graph attains the bound, implying that $h_2(n, n - 1) = n - 1$. For oriented graphs, Katona and Szemer\'{e}di~\cite{katona1967problem} showed the following.

\begin{theorem}[Katona and Szemer\'{e}di~\cite{katona1967problem}]
\label{thm:intro1}
    \[\dfrac{n}{2}\log_2 \dfrac{n}{2} \leq f_2(n) \leq n \lceil \log_2 n \rceil.\]
\end{theorem}

These bounds were further improved independently by the works of F\"{u}redi, Horak, Pareek and Zhu~\cite{furedi1998minimal}, and Kostochka, Luczak, Simonyi and Sopena\cite{kostochka1999minimum} which proved the following results.

\begin{theorem}[F\"{u}redi, Horak, Pareek and Zhu~\cite{furedi1998minimal}]
\label{thm:intro2}
    For any $n \geq 9$,
    \[(1 - o(1))n \log_2 n \leq f_2(n) \leq n\log_2 n - \dfrac{3}{2}n.\]
\end{theorem}

\begin{theorem}[Kostochka, Luczak, Simonyi and Sopena~\cite{kostochka1999minimum}]
\label{thm:intro3}
    For a fixed $d \geq 2$ and $n$ large enough
    \[n(\log_d n - 4\log_d \log_d n - 5) \leq f_d(n) \leq \lceil \log_d n\rceil(n - \lceil \log_d n\rceil).\]
\end{theorem}

Theorems~\ref{thm:intro2} and \ref{thm:intro3} both improved upon Theorem~\ref{thm:intro1} but arrived at the same conclusion which is as follows:

\begin{theorem}[F\"{u}redi, Horak, Pareek and Zhu~\cite{furedi1998minimal} and Kostochka, Luczak, Simonyi and Sopena~\cite{kostochka1999minimum}]
Let $f_2(n)$ be the minimum number of arcs in an absolute oriented clique of order $n$. Then,
\[\lim_{n \to \infty} \frac{f_2(n)}{n \log_2 n} = 1.\]
\end{theorem}

\subsection{Motivation and our contributions} 
Thus, the function $f_2(n)$ is asymptotically equal to $n \log_2 n$ and consequently, is asymptotically increasing. However, as the exact value of the function remains unknown, particular properties of $f_2(n)$ remain unknown to date. In fact, in one of the personal conversations of the fourth author with Sopena, it was detected that the seemingly easy question ``is the function $f_2(n)$ increasing?'' does not have a trivial answer. In this article, we answer that question by proving the following result.  

\begin{theorem}\label{thm:incr}
The function $f_2(n)$ is strictly increasing. 
\end{theorem}

Furthermore, we conjecture an exact value of $f_2(n)$ and 
prove an upper bound of $f_2(n)$ which matches our conjecture, and also improves the best known upper bound for the same. The conjectured value is given via a recurrence relation, and thus, for stating our conjecture, we would first like to present a recursive sequence $x_n$ in the following. 
        
\begin{defn}
    The initial values of the sequence $x_m$ of integers are given by $x_1 = 0, x_2 = 1, x_3 = 2, x_4 = 4, x_5 = 5$. For $m \geq 6$, the sequence $x_m$ is given by the following recurrence relation: 
\begin{equation}
\label{eq:xn}
    x_m = (m - 1) + 
        \begin{cases} 
        x_{\frac{m + 1}{2}} + x_{\frac{m - 3}{2}} &  \text{if } m \equiv 1 \pmod 4\\
        2 x_{\frac{m - 1}{2}} &  \text{if } m \equiv 3 \pmod 4\\
        x_{\frac{m}{2}} + x_{\frac{m - 2}{2}} & \text{ otherwise.}\\
        \end{cases}
\end{equation}
\end{defn}

\begin{conjecture}
\label{conj:main}
    For all positive integers $n$, we have $f_2(n) = x_n$.
\end{conjecture}

Our next contribution is to prove that $x_n$ serves as an upper bound of $f_2(n)$. In particular, this improves the best known upper bound of $f_2(n)$.

\begin{theorem}
\label{thm:upper}
    For all positive integers $n$, we have $f_2(n) \leq x_n$.
\end{theorem}

After the proofs, we present a discussion on why we think our conjecture is true, and some possible approaches to solve it.

\subsection{Organization}
In this article, first we recall some basic preliminaries in Section~\ref{sec:sec2}. After that, in Sections~\ref{sec:sec3} and \ref{sec:sec4}, we prove Theorems~\ref{thm:incr} and \ref{thm:upper}, respectively. In Section~\ref{sec:sec5}, we talk about the application of Theorem~\ref{thm:upper} to the oriented achromatic number of graphs. In Section~\ref{sec:sec6},
we conclude by including some discussions about our conjecture and provide some evidence in its support.

\medskip
\noindent \textbf{Note:} A preliminary version of this work was presented as a poster at EuroComb 2023.

\section{Preliminaries}\label{sec:sec2}
Given an arc $\overrightarrow{uv}$ 
of $\overrightarrow{G}$, the vertex $u$ is an \textit{in-neighbor} of $v$ and the vertex $v$ is an \textit{out-neighbor} of $u$. 
The set of all in-neighbors (resp., out-neighbors) of $u$ is denoted by $N^-(u)$ (resp., $N^+(u)$).
Two vertices $u$ and $v$ in  $\overrightarrow{G}$  \textit{agree} on a third vertex $w$ if $w$ is either an in-neighbor or an out-neighbor of both $u$ and $v$. To \textit{push} a vertex $u$ 
is to reverse the orientations of the arcs  incident to $u$. Note that, pushing $u$ swaps the 
sets $N^+(u)$ and $N^{-}(u)$.

We first examine the effect of pushing a vertex in an absolute oriented clique. For convenience, we say that a vertex $u$ \textit{sees} a vertex $v$ if $u$ and $v$ are either adjacent or connected by a directed 2-path. Furthermore, if $u$ is connected to $v$ by a directed $2$-path with $w$ being the internal vertex, then we say that $u$ sees $v$
\textit{through} $w$. 

\section{Proof of Theorem~\ref{thm:incr}}
\label{sec:sec3}
If we push a vertex $x$ of an oriented graph $\overrightarrow{H}$, then the adjacency between two vertices of the graph remains invariant (though the direction of the arc may change). Moreover, if there were a directed $2$-path through $x$ in $\overrightarrow{H}$, 
after pushing $x$ it remains a directed $2$-path, even though the directions of its arcs get reversed. 

\begin{observation}\label{lem:chiby1}
Let $\overrightarrow{H}$ be an absolute oriented clique
and $\overrightarrow{H}'$ be the oriented graph obtained by pushing a vertex $x$ of $\overrightarrow{H}$. Then the vertices of $V(\overrightarrow{H}') \setminus \{x\}$ see each other in $\overrightarrow{H}'$ and $n - 1 \leq \chi_o(\overrightarrow{H}') \leq n$. 
\end{observation}

 A refinement of the above observation follows. 

\begin{lemma}\label{lem x-x'}
Let $\overrightarrow{H}$ be an absolute oriented clique on $n$ vertices and let  $\overrightarrow{H}'$ be the oriented graph obtained by pushing the vertex $x$ of $\overrightarrow{H}$. If $\chi_o(\overrightarrow{H}') = n - 1$, then there exists a vertex $x'$, non-adjacent to $x$, which agrees with $x$ on all common neighbors in $\overrightarrow{H}'$. 
\end{lemma}

\begin{proof}
Let $\chi_o(\overrightarrow{H}') = n - 1$. Hence, there exists an oriented graph $\overrightarrow{T}$ of order $(n - 1)$, such that $\phi$ is a homomorphism of $\overrightarrow{H}'$ to $\overrightarrow{T}$. By Observation~\ref{lem:chiby1}, we know that all vertices in $V(\overrightarrow{H}') \setminus \{x\}$, see each other and hence, have distinct images under any homomorphism. Thus by the pigeonhole principle, we must have $\phi(x) = \phi(x')$ for some $x' \in V(\overrightarrow{H}') \setminus \{x\}$. 
  
As $\phi(x) = \phi(x')$, $x$ and $x'$ cannot be adjacent in $\overrightarrow{H}$, and consequently, in $\overrightarrow{H}'$. Therefore, there exists a directed 2-path between $x$ and $x'$ in the absolute oriented clique $\overrightarrow{H}$. 

Let  $w$ be a common neighbor of $x$ and $x'$. Note that, if $w$ is an in-neighbor (resp., out-neighbor) of $x$ in $\overrightarrow{H}'$, then $\phi(w)$ is an in-neighbor (resp., out-neighbor) of $\phi(x) = \phi(x')$ in $\overrightarrow{T}$. This implies that $w$ is an in-neighbor (resp., out-neighbor) of $x'$ too, as $\phi$ is a homomorphism. Thus, $x$ and $x'$ agree on every common neighbors in $\overrightarrow{H}'$.   
\end{proof}

Let $\overrightarrow{H}$ be an oriented graph having two non-adjacent vertices $x$ and $x'$ which agree on their common neighbors. Let $\overrightarrow{H}_{(xx')}$ be the graph obtained by identifying $x$ and $x'$. Formally, $\overrightarrow{H}_{(xx')}$ is the oriented graph obtained from $\overrightarrow{H}$ by deleting the vertices $x$ and $x'$ and then adding a new vertex $v_{xx'}$ having the union of in-neighbors of $x$ and $x'$ as its in-neighbors and the union of out-neighbors of $x$ and $x'$ as its out-neighbors. Then, we have the following.

\begin{lemma}
\label{lem:lem3}
Let $\overrightarrow{H}$ be an absolute oriented clique on $n$ vertices and let  $\overrightarrow{H}'$ be the oriented graph obtained by pushing the vertex $x$ of $\overrightarrow{H}$. If $\chi_o(\overrightarrow{H}') = n - 1$ and if the non-adjacent vertices $x, x'$ agree on every common neighbor in $\overrightarrow{H}'$, then the oriented graph $\overrightarrow{H}'_{(xx')}$ is an absolute oriented clique. 
\end{lemma}

\begin{proof}
To prove that $\overrightarrow{H}'' = \overrightarrow{H}_{(xx')}'$ is an absolute oriented clique, it is enough to show that the vertices of $\overrightarrow{H}''$ see each other due to Theorem~\ref{thm:occhar}. 

From Observation~\ref{lem:chiby1}, we know that the vertices of $V(\overrightarrow{H}') \setminus \{x\}$ see each other in $\overrightarrow{H}'$. Let $y,z \neq v_{xx'}$ be two distinct vertices in $\overrightarrow{H}''$. If $y$ and $z$ are adjacent in $\overrightarrow{H}'$, then they are adjacent in $\overrightarrow{H}''$ too. Suppose that $y$ sees $z$ through some vertex $w$ in $\overrightarrow{H}'$. If $w \neq x,x'$, then $y$ sees $z$ through $w$ in $\overrightarrow{H}''$ too. If $w \in \{x,x'\}$, then in $\overrightarrow{H}''$, $y$ sees $z$ through $v_{xx'}$. Thus, every vertex in $V(\overrightarrow{H}'') \setminus \{v_{xx'}\}$ sees every other vertex in $\overrightarrow{H}''$.

On the other hand, the fact that $v_{x x'}$ sees every vertex in $V(\overrightarrow{H}'') \setminus \{v_{x x'}\}$ is straightforward from the fact that $x'$ sees every vertex in $V(\overrightarrow{H}') \setminus \{x, x'\}$ and the fact that $N^+(x') \subseteq N^+(v_{xx'}), N^-(x') \subseteq N^-(v_{xx'})$.
\end{proof}

A vertex $x$ of an oriented graph $\overrightarrow{H}$ is a \textit{source} if $x$ has no in-neighbors.

\begin{lemma}\label{lem source del}
Let $\overrightarrow{H}$ be an absolute oriented clique on $n$ vertices and let $x$ be a source. Then the oriented graph $\overrightarrow{H} - x$ is also an absolute oriented clique. 
\end{lemma}

\begin{proof}
Every pair of vertices in $V(\overrightarrow{H}) \setminus \{x\}$ see each other, but not through $x$, as $x$ is a source. So all vertices in $\overrightarrow{H} - x$ see each other, and therefore, $\overrightarrow{H} - x$ is an absolute oriented clique.
\end{proof}

Finally we are ready to prove Theorem~\ref{thm:incr}. 

\bigskip

\noindent \textit{Proof of Theorem~\ref{thm:incr}.}
Let $\overrightarrow{H}$ be an absolute oriented clique of order $n + 1$ having exactly $f_2(n + 1)$ arcs for some $n \geq 1$. We wish to show the existence of an absolute oriented clique on $n$ vertices having strictly less than $f_2(n + 1)$ arcs. 

Let $w \in V(\overrightarrow{H})$ and $N^-(w) = \{w_1, w_2, \cdots, w_t\}$. We define a sequence $\overrightarrow{H}^i$   of oriented graphs, where $i \in \{1, 2, \cdots, t\}$. 
The oriented graph $\overrightarrow{H}^i$ is obtained by pushing the vertices of the set $\{w_1, w_2, \cdots, w_i\}$, and $\overrightarrow{H}^0 = \overrightarrow{H}$. Then the following two cases are possible.

\medskip
\noindent \textit{Case 1 - $\overrightarrow{H}^t$ is not an absolute oriented clique}: Then there exists an $i \in \{1, 2, \cdots, t\}$ such that $\overrightarrow{H}^0, \overrightarrow{H}^1, \cdots, \overrightarrow{H}^{i - 1}$ are all absolute oriented cliques while $\overrightarrow{H}^i$ is not. Since $\overrightarrow{H}^i$ is obtained by pushing a vertex $w_i$ in the oriented clique $\overrightarrow{H}^{i - 1}$, by Lemma~\ref{lem x-x'}, there exists a pair of non-adjacent vertices, say $x, x'$, in $\overrightarrow{H}^i$ which agree on all their common neighbors. Then, by Lemma~\ref{lem:lem3}, the oriented graph $\overrightarrow{H}^i_{(xx')}$ is an absolute oriented clique on $n$ vertices. Since $x$ and $x'$ see each other in $\overrightarrow{H}^{i - 1}$, they have at least one common neighbor in $\overrightarrow{H}^{i - 1}$. Therefore, $\overrightarrow{H}^i_{(xx')}$ is an absolute oriented clique of order $n$ and has size strictly less than $|A(\overrightarrow{H}^i)| = |A(\overrightarrow{H})| = f_2(n + 1)$. Hence, $f_2(n) < f_2(n + 1)$.

\medskip
\noindent \textit{Case 2 - $\overrightarrow{H}^t$ is an absolute oriented clique}: Since $w \in \overrightarrow{H}^t$ is a source, by Lemma~\ref{lem source del}, the oriented graph $\overrightarrow{H}^t - w$ is an absolute oriented clique on $n$ vertices. Since $deg(w) \geq 1$, the number of arcs in $\overrightarrow{H}^t - w$ is strictly less than $|A(\overrightarrow{H}^t)| = |A(\overrightarrow{H})| = f_2(n + 1)$. Therefore, $f_2(n) < f_2(n + 1)$.

\medskip
Thus, the function $f_2(n)$ is strictly increasing. \qed

\section{Proof of Theorem~\ref{thm:upper}}
\label{sec:sec4}
To prove this result, we will first describe the construction of a sequence of absolute oriented cliques $\overrightarrow{O}_n$ for all $n \geq 1$ in such a way that $\overrightarrow{O}_n$ has $x_n$ arcs (recall the sequence $((x_m))$ defined in Section~\ref{sec:sec1}). Observe that this is enough to prove Theorem~\ref{thm:upper}. Our proof consists of two parts: first we will describe the construction, and next we will show that our construction indeed satisfies $|A(\overrightarrow{O}_n)| = x_n$. For the second part of the proof, we also need to prove some properties of the sequence $((x_m))$. For the convenience of the reader, we will present our proof across a few subsections. 

\subsection{The construction of $\overrightarrow{O}_n$}
\label{cons:main}
The absolute oriented cliques $\overrightarrow{O}_n$, for $n \leq 5$, are depicted explicitly in Fig.~\ref{fig:smacliq}. One can easily verify that they are
absolute oriented cliques. Also, according to Sopena~\cite{sopena2014complete}, each of them has $f_2(n)$ many arcs, which in this case is the same as $x_n$ (the initial values).

\begin{figure}[t]
    \centering
\begin{tikzpicture}[scale=0.5]
% vertices
\tikzset{vertex/.style = {shape=circle,fill,inner sep=2pt,minimum size=3pt}}
\node[vertex] (x) at  (-13.5,0) {};
\node[vertex] (y) at  (-11.5,1.5) {};
\node[vertex] (z) at  (-11.5,-0.5) {};
\node[vertex] (w) at  (-9.5,2.5) {};
\node[vertex] (u) at  (-9.5,0.5) {};
\node[vertex] (v) at  (-9.5,-1.5) {};
\node[vertex] (a) at  (-7.5,0) {};
\node[vertex] (b) at  (-5.5,0) {};
\node[vertex] (c) at  (-3.5,-1.5) {};
\node[vertex] (d) at  (-3.5,1.5) {};
\node[vertex] (A) at  (18:2.5) {};
\node[vertex] (B) at  (90:2.5) {};
\node[vertex] (C) at  (162:2.5) {};
\node[vertex] (D) at  (234:2.5) {};
\node[vertex] (E) at  (306:2.5) {};
%edges
\draw[ultra thick,->] (z) -- (y);
\draw[ultra thick,->] (v) -- (u);
\draw[ultra thick,->] (u) -- (w);
\draw[ultra thick,->] (a) -- (b);
\draw[ultra thick,->] (b) -- (c);
\draw[ultra thick,->] (b) -- (d);
\draw[ultra thick,->] (c) -- (d);
\draw[ultra thick,->] (A) -- (B);
\draw[ultra thick,->] (B) -- (C);
\draw[ultra thick,->] (C) -- (D);
\draw[ultra thick,->] (D) -- (E);
\draw[ultra thick,->] (E) -- (A);
\node at (-13.5,-1) {$\overrightarrow{O}_1$};
\node at (-11.5,-1.5) {$\overrightarrow{O}_2$};
\node at (-9.5,-2.5) {$\overrightarrow{O}_3$};
\node at (-5.5,-2) {$\overrightarrow{O}_4$};
\node at (0,-3) {$\overrightarrow{O}_5$};
\end{tikzpicture}
    \caption{The absolute oriented cliques $\protect\overrightarrow{O}_1, \protect\overrightarrow{O}_2, \protect\overrightarrow{O}_3, \protect\overrightarrow{O}_4$, and $\protect\overrightarrow{O}_5$.}
    \label{fig:smacliq}
\end{figure}

Given two oriented graphs $\overrightarrow{G}_1$ and $\overrightarrow{G}_2$, the oriented graph $\overrightarrow{G}_1 \ltimes \overrightarrow{G}_2$ is obtained as follows:
we take the disjoint union of $\overrightarrow{G}_1$ and $\overrightarrow{G}_2$ and add a new vertex $v$ (say). Then, add arcs of the form $\overrightarrow{uv}$ for all vertices $u \in V(\overrightarrow{G}_1)$ and arcs of the form $\overrightarrow{vw}$ for all vertices 
$w \in V(\overrightarrow{G}_2)$. This construction is illustrated in Figure~\ref{fig:restricted}. A particularly interesting and useful property of $\overrightarrow{G}_1 \ltimes \overrightarrow{G}_2$ is the following that makes this construction important in the study of oriented coloring.

\begin{proposition}
\label{prop:crnoofG1+G2}
    Let $\overrightarrow{G}_1$ and $\overrightarrow{G}_2$ be two oriented graphs. Then 
    $\chi_o(\overrightarrow{G}_1 \ltimes \overrightarrow{G}_2) = \chi_o(\overrightarrow{G}_1) + \chi_o(\overrightarrow{G}_2) + 1$. 
\end{proposition}
\begin{proof}
    First note that the vertices of $\overrightarrow{G}_1$
    see the vertices of $\overrightarrow{G}_2$ through the special vertex $v$ (from the definition of $\overrightarrow{G}_1 \ltimes \overrightarrow{G}_2$). Thus the vertices of $\overrightarrow{G}_1$ and the vertices of $\overrightarrow{G}_2$ have distinct images under any homomorphism due to Theorem~\ref{thm:occhar}. Moreover, as $v$ is adjacent to every other vertex, it must have an image distinct 
    from every other vertex under any homomorphism. Therefore, we have $\chi_o(\overrightarrow{G}_1 \ltimes \overrightarrow{G}_2) \geq  \chi_o(\overrightarrow{G}_1) + \chi_o(\overrightarrow{G}_2) + 1$.

    On the other hand, there exists an oriented graph     $\overrightarrow{H}_i$ on $\chi_o(\overrightarrow{G}_i)$ vertices such that $\overrightarrow{G}_i$ admits a homomorphism to     $\overrightarrow{H}_i$, for $i \in \{1,2\}$. Thus, $\chi_o(\overrightarrow{G}_1 \ltimes \overrightarrow{G}_2)$ admits a homomorphism to $\chi_o(\overrightarrow{H}_1 \ltimes \overrightarrow{H}_2)$. As $\chi_o(\overrightarrow{H}_1 \ltimes \overrightarrow{H}_2)$ has exactly $\chi_o(\overrightarrow{G}_1) + \chi_o(\overrightarrow{G}_2) + 1$ vertices, we have $\chi_o(\overrightarrow{G}_1 \ltimes \overrightarrow{G}_2) \leq  \chi_o(\overrightarrow{G}_1) + \chi_o(\overrightarrow{G}_2) + 1$.
\end{proof}

If we choose $\overrightarrow{G}_1$ and $\overrightarrow{G}_2$ to be absolute oriented cliques, then a direct corollary of Proposition~\ref{prop:crnoofG1+G2} tells us that 
$\overrightarrow{G}_1 \ltimes \overrightarrow{G}_2$ is also an absolute oriented clique.

\begin{figure}[t]
    \centering
\begin{tikzpicture}
% vertices
\filldraw[black] (4,1) circle (3pt);
%edges
\draw[ultra thick,->] (0.75,2.5) -- (3.8,1.1);
\draw[ultra thick,->] (0.75,-0.5) -- (3.8,0.9);
\draw[ultra thick,->] (4.2,1.1) -- (7.25,2.5);
\draw[ultra thick,->] (4.2,0.9) -- (7.25,-0.5);
\draw (0,1) ellipse (1.25 and 2);
\draw (8,1) ellipse (1.25 and 2);
\node at (0,1) {$\overrightarrow{G}_1$};
\node at (8,1) {$\overrightarrow{G}_2$};
\node at (4,0.5) {$v$};
\end{tikzpicture}
    \caption{The oriented graph $\overrightarrow{G}_1 \ltimes \overrightarrow{G}_2$.}
    \label{fig:restricted}
\end{figure}

\begin{corollary}
\label{cor:butgraphisoclique}
   Let $\overrightarrow{G}_1$ and $\overrightarrow{G}_2$ be two absolute oriented cliques, then the oriented graph $\overrightarrow{G}_1 \ltimes \overrightarrow{G}_2$
   is also an absolute oriented clique. 
\end{corollary}
\begin{proof}
    Let $\overrightarrow{G}_1$ and $\overrightarrow{G}_2$ be absolute oriented cliques, that is, $\chi_o(\overrightarrow{G}_1) = |V(\overrightarrow{G}_1)|$ and $\chi_o(\overrightarrow{G}_2) = |V(\overrightarrow{G}_2)|$. We know that $\chi_o(\overrightarrow{G}_1 \ltimes \overrightarrow{G}_2) = \chi_o(\overrightarrow{G}_1) + \chi_o(\overrightarrow{G}_2) + 1 = |V(\overrightarrow{G}_1)| + |V(\overrightarrow{G}_2)| + 1$ due to Proposition~\ref{prop:crnoofG1+G2}, which is precisely the number of vertices in $\overrightarrow{G}_1 \ltimes \overrightarrow{G}_2$. Hence, by definition, the oriented graph $\overrightarrow{G}_1 \ltimes \overrightarrow{G}_2$ is also an absolute oriented clique.
\end{proof}

In particular, Proposition~\ref{prop:crnoofG1+G2} also proves a recursive upper bound for the function $f_2(n)$.

\begin{corollary}
    For any $n \geq 3$, we have 
    $$f_2(n) \leq (n-1) + \min\{f_2(n_1)+f_2(n_2) : n_1+n_2 = n-1\}.$$
\end{corollary}
\begin{proof}
    Let $\overrightarrow{G}_{n_1}$ and $\overrightarrow{G}_{n_2}$ be absolute oriented cliques of order $n_1$ and $n_2$ respectively. Let $|A(\overrightarrow{G}_{n_1})| = f_2(n_1)$ and $|A(\overrightarrow{G}_{n_2})| = f_2(n_2)$. If $\overrightarrow{G}_n = \overrightarrow{G}_{n_1} \ltimes \overrightarrow{G}_{n_2}$, then by the construction, $n_1 + n_2 = n - 1$, and by Corollary~\ref{cor:butgraphisoclique}, $\overrightarrow{G}_n$ is an absolute oriented clique. Hence, the number of arcs of $\overrightarrow{G}_n$ is an upper bound for $f_2(n)$. Now we can vary the values of $n_1$ and $n_2$, and choose the values where we get the minimum possible bound of $f(n)$ in the 
    above-mentioned process. This gives us the desired result. 
\end{proof}

Intuitively, we feel that the upper bound presented in the above corollary is actually tight. Our construction is based on this intuition. Recall that 
the initial $\overrightarrow{O}_n$'s, that is, $\overrightarrow{O}_n$ for $n \leq 5$, were depicted explicitly in Fig.~\ref{fig:smacliq}. 
For $n \geq 6$, we define 
$$\overrightarrow{O}_n = \overrightarrow{O}_{n_1} \ltimes \overrightarrow{O}_{n_2}$$ 
where $n_1 + n_2 = n - 1$ and $n_1$ and $n_2$ are chosen such that $|A(\overrightarrow{O}_{n_1})| + |A(\overrightarrow{O}_{n_2})|$ is minimum.
Our next task is to show that the number of arcs in 
$\overrightarrow{O}_n$ is indeed equal to $x_n$ for all 
$n \geq 1$. We will establish this part of the proof in the subsequent sections.

\subsection{Properties of the sequence $((x_m))$}
\begin{lemma}
\label{lem:x_nprop1}
For all $m \geq 3$, $$x_{m + 1} = \begin{cases}
        x_{m} + x_{m - 1} - x_{m - 2} + 1, & \text{if }m = 3 \cdot 2^i \text{ or } 3 \cdot 2^i - 1, \text{ for some } i \geq 0,\\
        x_{m} + x_{m - 1} - x_{m - 2}, & \text{otherwise.}
    \end{cases}$$
\end{lemma}
\begin{proof}
We prove this using strong induction on $m$, that is, suppose that the identity stated above is true for all $m \in \{3, 4, 5, \ldots, k - 1\}$. 
However, in hindsight we know that the proof of the induction step will work uniformly for all $m = k \geq 6$. Therefore, even though it is enough to verify the statement for $m = 3$ to prove the base case, we will verify the statement for all $m = 3, 4, 5$ which will serve as our base case. 

For the base cases, we will verify the correctness of the given identity for $m \in \{3,4,5\}$. As these are just straightforward calculations, we simply list them below for the convenience of the readers. 

\begin{itemize}
    \item \textit{For $m = 3$:} $x_{3} + x_{2} - x_{1} + 1 = 2 + 1 - 0 + 1 = 4 = x_{4}$.
    \item \textit{For $m = 4$:} $x_{4} + x_{3} - x_{2} = 4 + 2 - 1 = 5 = x_{5}$.
    \item \textit{For $m = 5$:} $x_{5} + x_{4} - x_{3} + 1 = 5 + 4 - 2 + 1 = 8 = x_{6}$.
\end{itemize}

Next, we are going to prove the induction step, that is, assuming the identity from the statement is true for all $m \leq k - 1$, we will prove that the identity holds for $m = k$. As the identity depends on the residue modulo $4$ value of $k$, we need to consider some cases.

\medskip
\noindent \textit{Case 1 - When $k \equiv 1,2 \pmod 4$}: 
In this case, $k$ cannot be of the form $3 \cdot 2^i$ or $3 \cdot 2^i - 1$, where $i$ is an integer greater than or equal to $2$. Hence, the identity can be established by straightforward substitution of the appropriate terms from eq$^n$(\ref{eq:xn}). We provide the detailed calculations below. 

\medskip

\begin{itemize}
    \item If $k \equiv 1 \pmod 4$, then
\begin{align*}
    x_{k} + x_{k - 1} - x_{k - 2} & = \left(k - 1 + x_{\frac{k + 1}{2}} + x_{\frac{k - 3}{2}}\right) + \left(k - 2 + x_{\frac{k - 1}{2}} + x_{\frac{k - 3}{2}}\right)\\
    & \hspace{6cm} - \left(k - 3 + 2x_{\frac{k - 3}{2}}\right)\\ & = k + x_{\frac{k + 1}{2}} + x_{\frac{k - 1}{2}}\\
    & = x_{k + 1}.
\end{align*}

\item If $k \equiv 2 \pmod 4$, then
\begin{align*}
    x_{k} + x_{k - 1} - x_{k - 2} & = \left(k - 1 + x_{\frac{k}{2}} + x_{\frac{k - 2}{2}}\right) + \left(k - 2 + x_{\frac{k}{2}} + x_{\frac{k - 4}{2}}\right)\\ 
    & \hspace{6cm} - \left(k - 3 + x_{\frac{k - 2}{2}} + x_{\frac{k - 4}{2}}\right)\\
    & = k + 2 x_{\frac{k}{2}}\\
    & = x_{k + 1}.
\end{align*}
\end{itemize}

\medskip
\noindent \textit{Case 2 - When $k \equiv 0,3 \pmod 4$}: 
In this case, we need to take into account those $k$'s which are of the forms $3 \cdot 2^i$ or $3 \cdot 2^i - 1$. The key observation in this case is that when $k \geq 6$, the integer $k$ is of the form $3\cdot2^i - 1$ if and only if $(k - 1)/2$ is of the same form. So the identity for both indices $k$ and $(k - 1)/2$ has the $+1$ term, and the substitution of appropriate terms in the eq$^n$(\ref{eq:xn}) gives the required expression. Similarly, when $k \geq 6$, the integer $k$ is of the form $3\cdot2^i$ if and only if $k/2$ is of the same form. So the identity for both indices $k$ and $k/2$ has the $+1$ term, and the substitution of appropriate terms in eq$^n$(\ref{eq:xn}) 
gives the required expression. We provide the detailed 
calculations below. The term $\epsilon$, used below for convenience, takes the value $1$ if $k$ is of the form $3 \cdot 2^i$ or $3 \cdot 2^i - 1$, and $0$ otherwise.

\medskip
\begin{itemize}
    \item If $k \equiv 3 \pmod 4$, then
    \begin{align*}
    x_{k} + x_{k - 1} - x_{k - 2} + \epsilon & = \left(k - 1 + 2 x_{\frac{k - 1}{2}}\right) + \left(k - 2 + x_{\frac{k - 1}{2}} + x_{\frac{k - 3}{2}}\right)\\
    & \hspace{5cm} - \left(k - 3 + x_{\frac{k - 1}{2}} + x_{\frac{k - 5}{2}}\right) + \epsilon\\
    & = k + \left(x_{\frac{k - 1}{2}} + x_{\frac{k - 3}{2}} - x_{\frac{k - 5}{2}}\right) + x_{\frac{k - 1}{2}} + \epsilon\\
    & = k + x_{\frac{k + 1}{2}} - \epsilon + x_{\frac{k - 1}{2}} + \epsilon\\
    & = x_{k + 1}.
\end{align*}
    
    \item If $k \equiv 0 \pmod 4$, then
\begin{align*}
    x_{k} + x_{k - 1} - x_{k - 2} + \epsilon & = \left(k - 1 + x_{\frac{k}{2}} + x_{\frac{k - 2}{2}}\right) + \left(k - 2 + 2 x_{\frac{k - 2}{2}}\right)\\
    & \hspace{5cm} - \left(k - 3 + x_{\frac{k - 2}{2}} + x_{\frac{k - 4}{2}}\right) + \epsilon\\ & = k + \left(x_{\frac{k}{2}} + x_{\frac{k - 2}{2}} - x_{\frac{k - 4}{2}} \right) + x_{\frac{k - 2}{2}} + \epsilon\\ & = k + x_{\frac{k + 2}{2}} - \epsilon + x_{\frac{k - 2}{2}} + \epsilon\\
    & = x_{k + 1}.
\end{align*}  
\end{itemize}

\noindent

The above conclusion proves the induction step, and thus completes the proof of the lemma. 
\end{proof}

Having seen one property of the terms of the sequence $((x_m))$ let us motivate the results in this section. Recall that the construction of the absolute oriented clique $\overrightarrow{O}_n$ requires two smaller absolute oriented cliques of orders $n_1$ and $n_2$ satisfying $n_1 + n_2 = n - 1$ such that $|A(\overrightarrow{O}_{n_1})| + |A(\overrightarrow{O}_{n_2})|$ is minimum. It is not clear by the definition, for what values of $n_1$ and $n_2$ this minimum is attained. We showed that the initial values ($n \leq 5$) of $x_n$ matches with $|A(\overrightarrow{O}_n)|$. To eventually show that $|A(\overrightarrow{O}_n)| = x_n$, we use the precise values of $x_n$ that are known (from eq$^n$(\ref{eq:xn})), to determine the values of possible $n_1$'s and $n_2$'s that minimizes the sum $x_{n_1} + x_{n_2}$. In fact, we claim that the optimal values of $n_1$ and $n_2$ is roughly around $\frac{n}{2}$. To show this, we divide our arguments into two cases: $n$ being even and $n$ being odd. For each of these cases, we prove some more properties of the sequence $((x_m))$ which will support our claim.

\subsubsection{$n = 2m$ is even}
Suppose $n$ is even and of the form $n = 2m$. The unordered pairs $\{1, 2m -2\}, \{2, 2m - 3\}, \cdots,\{m - 2, m + 1\}$ and $\{m - 1, m\}$ give us possible choices for the values of $n_1$ and $n_2$ such that $n_1 + n_2 = n - 1 = 2m - 1$. We compare the values of $x_{n_1} + x_{n_2}$ for each of these pairs, and determine for which values, the sum is minimized. For example, by rearranging the terms in the equality given in Lemma~\ref{lem:x_nprop1}, we get the following inequality as a direct corollary, which tells us that the choice of $\{n_1, n_2\} = \{m - 1, m\}$ is better than $\{n_1, n_2\} = \{m - 2, m + 1\}$.

\begin{corollary}
\label{cor:x_nprop1}
For $m \geq 3$, we have $x_{m} + x_{m - 1} \leq x_{m + 1} + x_{m - 2}$.
\end{corollary}

In the following lemma, we complete the discussion for even $n$ by establishing that the choice of $\{n_1, n_2\} = \{m - 1, m\}$ is, in fact, the best possible among all the options for $n_1$ and $n_2$.

\begin{lemma}
\label{lem:x_nprop2}
    For $m \geq 2$, $x_{m - i} + x_{m - 1 + i} \leq x_{m + i} + x_{m - 1 - i}$, where $i \in [0, m - 2]$. Moreover, $x_{m} + x_{m - 1} \leq x_{m + i} + x_{m - 1 - i}$, where $i \in [0, m - 2]$.
\end{lemma}
\begin{proof}
    Fix an $m \geq 2$. To prove the first part of the Lemma, we use strong induction on the index $i$, that is suppose that the identity stated above is true for all $i \in \{0, 1, \cdots, j - 1\}$. To prove the base case, we verify the statement for $i = 0$ and $1$, since substituting $i = 0$, gives us a trivial result. By substituting $i = 1$ in the identity, we get $x_{m - 1} + x_{m} \leq x_{m + 1} + x_{m - 2}$ which is true due to Corollary~\ref{cor:x_nprop1}. This proves the base case.
    
    Next, we are going to prove the induction step, that is, assuming the identity from the statement is true for all $i \leq j - 1$, we will prove that the identity holds for $i = j$. To do this, we start with $x_{m + j} + x_{m - 1 - j}$, and apply Corollary~\ref{cor:x_nprop1} so that the terms obtained can be grouped together. After that we will apply the induction hypothesis to get the required proof. We provide the detailed calculations below.
    
    We replace $m$ with $(m - 1 + j)$ and $(m + 1 - j)$, respectively, in Corollary~\ref{cor:x_nprop1} and make some rearrangements to get
\begin{align*}
    x_{m + j} + x_{m - 1 - j} & \geq (x_{m - 1 + j} + x_{m - 2 + j} - x_{m - 3 + j}) + (x_{m + 1 - j} + x_{m - j} - x_{m + 2 - j})\\ & = (x_{m - 1 + j} + x_{m - j}) + (x_{m - 2 + j} + x_{m + 1 - j}) - (x_{m - 3 + j} +  x_{m + 2 - j}).
\end{align*}
But from the induction hypothesis, when $i = j - 2$ $$x_{m - 3 + j} + x_{m + 2 - j} \leq x_{m - 2 + j} + x_{m + 1 - j}.$$ 
Hence, by combining the above two inequalities we have 
\begin{align*}
    x_{m + j} + x_{m - 1 - j} & \geq (x_{m - 1 + j} + x_{m - j}) + (x_{m - 3 + j} +  x_{m + 2 - j}) - (x_{m - 3 + j} +  x_{m + 2 - j})\\
    & = x_{m - 1 + j} + x_{m - j},
\end{align*}
which completes the proof of the induction step. 

The second half of the lemma can be proved by repeatedly applying the identity that we have just proved. The specific calculations are as follows.
\begin{align*}
    x_{m + i} + x_{m - 1 - i} \geq x_{m - 1 + i} + x_{m - i} & = x_{m + (i - 1)} + x_{m - 1 - (i - 1)} \\
    & \geq x_{m + (i - 2)} + x_{m - 1 - (i - 2)}\\
    & \cdot\\
    & \cdot\\
    & \cdot\\
    & \geq x_{m + (1)} + x_{m - 1 - (1)}\\
    & \geq x_{m} + x_{m - 1}.
\end{align*}
This completes the proof of the lemma.
\end{proof}

\noindent \textbf{Remark:} Lemma~\ref{lem:x_nprop2} implies that the value of $x_{n_1} + x_{n_2}$ increases with the difference between $n_1$ and $n_2$, when $n_1 + n_2 = n - 1$. This supports our claim that the best choice for $n_1$ and $n_2$ is roughly around $\frac{n}{2}$, where the difference between $n_1$ and $n_2$ is the minimum. The second half of Lemma~\ref{lem:x_nprop2} proves our claim.

\subsubsection{$n = 2m + 1$ is odd}
Suppose $n$ is odd and of the form $n = 2m + 1$. The unordered pairs $\{1, 2m -1\}, \{2, 2m - 2\}, \cdots, \{m - 1, m+ 1\}$ and $\{m, m\}$ give us possible choices for the values of $n_1$ and $n_2$ such that $n_1 + n_2 = n - 1 = 2m$. Just as before, we compare the values of $x_{n_1} + x_{n_2}$ for each of these pairs, and determine for which values, the sum is minimized. 

The following lemma establishes the fact that when $n$ is odd, among all the choices for $n_1$ and $n_2$, the best choice is either $\{n_1, n_2\} = \{m - 1, m + 1\}$ or $\{n_1, n_2\} = \{m, m\}$.
\begin{lemma}
\label{lem:x_nprop3}
For $m \geq 2$, $x_{m + i - 1} + x_{m - i + 1} \leq x_{m + i + 1} + x_{m - i - 1}$, where $i \in [0, m - 2]$. Moreover,
    $$x_{m + i + 1} + x_{m - i - 1} \geq 
        \begin{cases}
            x_{m - 1} + x_{m + 1}, & i \text{ is even,}\\
            2x_{m}, & i \text{ is odd,}
        \end{cases}$$
    where $i \in [-1, m - 2]$.
\end{lemma}
\begin{proof}
    Fix an $m \geq 2$. To prove the first part of the lemma, we use strong induction on the index $i$, that is, suppose that the identity stated above is true for all $i \in \{0, 1, 2, \ldots, j - 1\}$. By substituting $i = 0$ in the identity, we get a trivial result which proves the base case. 
    
    Next, we are going to prove the induction step, that is, assuming the identity from the statement is true for all $i \leq j - 1$, we will prove that the identity holds for $i = j$. To do this, we start with $x_{m + j + 1} + x_{m - j - 1}$, and apply Corollary~\ref{cor:x_nprop1} so that the terms obtained can be grouped together. After that, we will apply the induction hypothesis to get the required proof. We provide the detailed calculations below.

    Replace $m$ by $m + j$ and $m$ by $m - j + 1$ respectively in Corollary~\ref{cor:x_nprop1} to get
    \begin{align*}
        x_{m + j + 1} + x_{m - j - 1} & \geq x_{m + j} + x_{m + j - 1} - x_{m + j - 2} + x_{m - j + 1} + x_{m - j} - x_{m - j + 2}\\
        & = (x_{m + j} + x_{m - j}) + (x_{m + j - 1} + x_{m - j + 1}) - (x_{m + j - 2} + x_{m - j + 2})
    \end{align*}
    By applying the induction hypothesis for $i = j - 1$, we get $$x_{m + j} + x_{m - j} \geq x_{m + j - 2} + x_{m - j + 2}.$$ 
    Hence, we have,
    \begin{align*}
        x_{m + j + 1} + x_{m - j - 1} & \geq (x_{m + j - 2} + x_{m - j + 2}) + (x_{m + j - 1} + x_{m - j + 1}) - (x_{m + j - 2} + x_{m - j + 2})\\
        & = x_{m + j - 1} + x_{m - j + 1},
    \end{align*}
    which completes the proof of the induction step.

    The second half of the lemma can be proved by recursively applying the identity that we have just proved. The calculations vary slightly depending on whether $i$ is even or odd. The specific calculations are as follows. 
    
    \begin{itemize}
        \item When $i$ is even:
        \begin{align*}
            x_{m + i + 1} + x_{m - 1 - i} \geq x_{m + i - 1} + x_{m - i + 1} & = x_{m + 1 + (i - 2)} + x_{m - 1 - (i - 2)} \\
            & \geq x_{m + 1 + (i - 4)} + x_{m - 1 - (i - 4)}\\
            & \cdot\\
            & \cdot\\
            & \cdot\\
            & \geq x_{m + 1 + (i - i)} + x_{m - 1 - (i - i)}\\
            & = x_{m + 1} + x_{m - 1}.
        \end{align*}
        \item When $i$ is odd:
        \begin{align*}
            x_{m + i + 1} + x_{m - 1 - i} \geq x_{m + i - 1} + x_{m - i + 1} & = x_{m + 1 + (i - 2)} + x_{m - 1 - (i - 2)} \\
            & \geq x_{m + 1 + (i - 4)} + x_{m - 1 - (i - 4)}\\
            & \cdot\\
            & \cdot\\
            & \cdot\\
            & \geq x_{m + 1 + (i - \{i + 1\})} + x_{m - 1 - (i - \{i + 1\})}\\
            & = 2 x_{m}.
        \end{align*}
    \end{itemize}
    This completes the proof of the lemma.
\end{proof}

\noindent \textbf{Remark:} Just as before, Lemma~\ref{lem:x_nprop3} implies that the value of $x_{n_1} + x_{n_2}$ increases with the difference between $n_1$ and $n_2$, when $n_1 + n_2 = n - 1$. This proves our claim that the best choice for $n_1$ and $n_2$ is roughly around $\frac{n}{2}$, where the difference between $n_1$ and $n_2$ is the minimum. 

The following two lemmas helps us to complete our understanding of the case when $n$ is odd by giving us conditions when $\{n_1, n_2\} = \{m - 1, m + 1\}$ is and when $\{n_1, n_2\} = \{m, m\}$ is a better choice for $n_1$ and $n_2$.

\begin{lemma}
\label{lem:x_nprop41}
    If $m \geq 2$ is even, then $2x_m \geq x_{m - 1} + x_{m + 1}$. 
\end{lemma}
\begin{proof}
    Let $m$ be even. The proof is by strong induction on $m$, that is, suppose that the identity stated above is true for all $m \in \{2, 4, \cdots, k - 2\}$. To prove the base case, it is enough to observe that when $m = 2$, $$x_1 + x_3 = 0 + 2 = 2 = 2x_2.$$ 
    
    Next, we are going to prove the induction step, that is, assuming the identity from the statement is true for all $m \leq k - 2$, we will prove that the identity holds for $m = k$. To do this, we will consider $x_{k - 1} + x_{k + 1}$ and substitute from eq$^n$(\ref{eq:xn}). Since eq$^n$(\ref{eq:xn}) depends on the residue modulo 4 value of $k$, we get two different substitutions depending on the value of $k$. In both cases, next we apply the induction hypothesis to simplify the terms and get the required result. We give detailed calculations below.

    \begin{itemize}
        \item If $k \equiv 0 \pmod 4$, then $$x_{k + 1} + x_{k - 1} = \left(k + x_{\frac{k + 2}{2}} + x_{\frac{k - 2}{2}}\right) + \left(k - 2 + 2 x_{\frac{k - 2}{2}}\right).$$ But since, $\frac{k}{2}$ is also even, the induction hypothesis for $m = \frac{k}{2}$ gives us that $$2x_{\frac{k}{2}} \geq x_{\frac{k}{2} - 1} + x_{\frac{k}{2} + 1} = x_{\frac{k - 2}{2}} + x_{\frac{k + 2}{2}},$$ and upon substituting and rearranging, we have 
    \begin{align*}
        x_{k + 1} + x_{k - 1} & \leq \left(k - 1 + x_{\frac{k - 2}{2}} + x_{\frac{k}{2}}\right) + \left(k - 1 + x_{\frac{k - 2}{2}} + x_{\frac{k}{2}}\right)\\ & = 2 x_k.
    \end{align*} 

        \item If $k \equiv 2 \pmod 4$, then $$x_{k + 1} + x_{k - 1} = \left(k + 2 x_{\frac{k}{2}}\right) + \left(k - 2 + x_{\frac{k}{2}} + x_{\frac{k - 4}{2}}\right).$$ But since, $\frac{k - 2}{2}$ is also even, the induction hypothesis for $m = \frac{k - 2}{2}$ gives us that $$2x_{\frac{k - 2}{2}} \geq x_{\frac{k - 2}{2} - 1} + x_{\frac{k - 2}{2} + 1} = x_{\frac{k - 4}{2}} + x_{\frac{k}{2}},$$ and upon substituting and rearranging, we have 
    \begin{align*}
        x_{k + 1} + x_{k - 1} & \leq \left(k - 1 + x_{\frac{k - 2}{2}} + x_{\frac{k}{2}}\right) + \left(k - 1 + x_{\frac{k - 2}{2}} + x_{\frac{k}{2}}\right)\\ & = 2 x_k.
    \end{align*}
    \end{itemize}
    This completes the proof of the lemma.
\end{proof}

\begin{lemma}
\label{lem:x_nprop42}
    If $m \geq 3$ is odd, then $2x_m \leq x_{m - 1} + x_{m + 1}$.
\end{lemma}
\begin{proof}
    Let $m$ be odd. We proceed similarly as above. The proof is by strong induction on the index $m$, that is, suppose that the identity stated above is true for all $m \in \{3, \cdots, k - 2\}$. To prove the base case, it is enough to observe that when $m = 3$, $$x_2 + x_4 = 1 + 4 = 5 \geq 4 = 2x_3.$$ 
    
     Next, we are going to prove the induction step, that is, assuming the identity from the statement is true for all $m \leq k - 2$, we will prove that the identity holds for $m = k$. To do this, we will consider $x_{k - 1} + x_{k + 1}$ and substitute from eq$^n$(\ref{eq:xn}). Since eq$^n$(\ref{eq:xn}) depends on the residue modulo 4 value of $k$, we get two different substitutions depending on the value of $k$. In both cases, depending on whether $\frac{k - 1}{2}$ is even or odd, we make appropriate substitutions to simplify the terms and get the required result. We give detailed calculations below.

    \begin{itemize}
        \item If $k \equiv 1 \pmod 4$, then $$x_{k + 1} + x_{k - 1} = \left(k + x_{\frac{k + 1}{2}} + x_{\frac{k - 1}{2}}\right) + \left(k - 2 + x_{\frac{k - 1}{2}} + x_{\frac{k - 3}{2}}\right).$$ Observe that $\frac{k - 1}{2}$ is even. Hence, by Lemma~\ref{lem:x_nprop41}, we have, $$2x_{\frac{k - 1}{2}} \geq x_{\frac{k - 1}{2} - 1} + x_{\frac{k - 1}{2} + 1} = x_{\frac{k - 3}{2}} + x_{\frac{k + 1}{2}}.$$ Therefore, by rearranging and grouping the terms, we get,
    \begin{align*}
        x_{k + 1} + x_{k - 1} & \geq \left(k - 1 + x_{\frac{k + 1}{2}} + x_{\frac{k - 3}{2}}\right) + \left(k - 1 + x_{\frac{k + 1}{2}} + x_{\frac{k - 3}{2}}\right)\\ & = 2 x_k.
    \end{align*}

        \item If $k \equiv 3 \pmod 4$, then $$x_{k + 1} + x_{k - 1} = \left(k + x_{\frac{k + 1}{2}} + x_{\frac{k - 1}{2}}\right) + \left(k - 2 + x_{\frac{k - 1}{2}} + x_{\frac{k - 3}{2}}\right).$$ But since, $\frac{k - 1}{2}$ is also odd, the induction hypothesis for $m = \frac{k - 1}{2}$ gives us $$2x_{\frac{k - 1}{2}} \leq x_{\frac{k - 1}{2} - 1} + x_{\frac{k - 1}{2} + 1} = x_{\frac{k - 3}{2}} + x_{\frac{k + 1}{2}},$$ and we have 
    \begin{align*}
        x_{k + 1} + x_{k - 1} & \geq \left(k - 1 + 2 x_{\frac{k - 1}{2}}\right) + \left(k - 1 + 2 x_{\frac{k - 1}{2}}\right)\\ & = 2 x_k.
    \end{align*}
    \end{itemize}
    This completes the proof of the induction step and hence, completes the proof of the lemma.
\end{proof}

\subsection{Concluding the proof of Theorem~\ref{thm:upper}}

\noindent \textit{Proof of Theorem~\ref{thm:upper}:} Now we finally have all the tools necessary to prove our result. As mentioned before, we prove this theorem by proving that for all $n \geq 1$, $|A(\overrightarrow{O}_n)| = x_n$.

    The proof is by strong induction on $n$, that is, suppose that the statement is true for all $n \in \{1, 2, \cdots, m - 1\}$. To prove the base case, we verify the statement for $n \leq 5$. This is true since we have already shown that the initial values $(n \leq 5)$ of $x_n$ match with the initial values of $|A(\overrightarrow{O}_n)|$. 
    
    Next, we are going to prove the induction step, that is, assuming the statement is true for all $n \leq m - 1$, we will prove that the statement holds for $n = m$. To do this, depending on the residue modulo 4 value of $m$, we will use the lemmas that we have proved to obtain the values of $n_1$ and $n_2$ through which the sum $x_{n_1} + x_{n_2}$ and equivalently $|A(\overrightarrow{O}_{n_1})| + |A(\overrightarrow{O}_{n_2})|$, is minimized. The details are listed below. For indices lesser than or equal to $m - 1$, the induction hypothesis allows us to use $|A(\overrightarrow{O}_m)|$ and $x_m$ interchangeably.

    \begin{itemize}
        \item When $m$ is odd: By Lemma~\ref{lem:x_nprop3}, $|A(\overrightarrow{O}_{n_1})| + |A(\overrightarrow{O}_{n_2})|$ is minimized either when $\{n_1, n_2\} = \{\frac{m - 1}{2} - 1, \frac{m - 1}{2} + 1\}$ or when $\{n_1, n_2\} = \{\frac{m - 1}{2}, \frac{m - 1}{2}\}$ depending on whether $\frac{m - 1}{2}$ is even or odd.
        \begin{itemize}
            \item When $m \equiv 1 \pmod 4$, $\frac{m - 1}{2}$ is even, hence, by Lemma~\ref{lem:x_nprop41}, \begin{align*}
                |A(\overrightarrow{O}_m)| & = m - 1 + |A(\overrightarrow{O}_{\frac{m - 3}{2}})| + |A(\overrightarrow{O}_{\frac{m + 1}{2}})|\\
                & = x_m.
            \end{align*}
            \item When $m \equiv 3 \pmod 4$, $\frac{m - 1}{2}$ is odd, hence, by Lemma~\ref{lem:x_nprop42}, \begin{align*}
                |A(\overrightarrow{O}_m)| & = m - 1 + 2|A(\overrightarrow{O}_{\frac{m - 1}{2}})|\\
                & = x_m.
            \end{align*}
        \end{itemize}
        \item When $m$ is even: By Lemma~\ref{lem:x_nprop2}, $|A(\overrightarrow{O}_{n_1})| + |A(\overrightarrow{O}_{n_2})|$ is minimized when $\{n_1, n_2\} = \{\frac{m}{2} - 1, \frac{m}{2}\}$. Hence, we have 
        \begin{align*}
                |A(\overrightarrow{O}_m)| & = m - 1 + |A(\overrightarrow{O}_{\frac{m - 2}{2}})| + |A(\overrightarrow{O}_{\frac{m}{2}})|\\
                & = x_m.
            \end{align*}
    \end{itemize}
    This completes the proof of the theorem.\qed

\section{Applications to the oriented achromatic number}
\label{sec:sec5}
Let $\overrightarrow{G}$ and $\overrightarrow{H}$ be oriented graphs. A surjective homomorphism $h: V(\overrightarrow{G}) \rightarrow V(\overrightarrow{H})$ is \textit{complete} if and only if for every arc $\overrightarrow{ab} \in A(\overrightarrow{H})$, there exists an arc $\overrightarrow{uv} \in A(\overrightarrow{G})$ such that $h(u) = a$ and $h(v) = b$. The \textit{oriented achromatic number} of an oriented graph $\overrightarrow{G}$, denoted $\psi_o(\overrightarrow{G})$, is the largest order of an absolute oriented clique $\overrightarrow{O}_n$ such that there is a complete homomorphism $h: V(\overrightarrow{G}) \rightarrow V(\overrightarrow{O}_n)$. The oriented achromatic number was first defined by Sopena~\cite{sopena2014complete}. 

By the definition, it is easy to see that the results on absolute oriented cliques will directly effect results on the oriented achromatic number. As an illustration, we give a few results that are improvements of results already proved in~\cite{PD202448}. An interested reader can refer to \cite{sopena2014complete} and \cite{PD202448} for the original results and conventions used. The sequence $((x_m))$ is as defined in eq$^n$(~\ref{eq:xn}).

Let $v_{odd}(G)$ be the number of vertices with odd degree in $G$. The following result is an improvement of Theorem 2.6 in \cite{PD202448}.
\begin{theorem}
\label{thm:path}
    For every integer $m \geq 2$ and every $n \geq \frac{k}{2} + x_m$, $\psi_o(P_n) \geq m$, where $$k = \begin{cases}
    v_{odd}(O_m) - 2, & \text{ if } m \neq 5\\
    v_{odd}(O_m), & \text{ if } m = 5. \end{cases}$$
\end{theorem}

Similarly, for cycles, we can improve Theorem 2.8 in \cite{PD202448} to get the following.
\begin{theorem}
    For every $m \geq 2$ and every $n \geq \frac{v_{odd}(O_m)}{2} + x_m$, $\psi_o(C_{n}) \geq m$.
\end{theorem}

The Cartesian product $G \boxempty H$ of graphs $G$ and $H$ is the graph with vertex set the Cartesian product $V(G) \times V(H)$, and, two vertices $(u, u')$ and $(v, v')$ are adjacent in $G \boxempty H$ if and only if either (a) $u = v$ and $u'$ is adjacent to $v'$ in $H$, or (b) $u' = v'$ and $u$ is adjacent to $v$ in $G$.

The next few results talk about the oriented achromatic number of cartesian products of certain types of graphs. We improve Theorems 3.11, 3.12 and 3.13 in \cite{PD202448} to get the following.

\begin{theorem}
\label{thm:cartpath}
    Let $n \geq 1$ be an integer. Then $\psi_o(P_n \boxempty K_2) < m$ if $3n + 1 < x_m$.
\end{theorem}

\begin{theorem}
\label{thm:cartcyc}
     Let $n \geq 3$ be an integer. Then $\psi_o(C_n \boxempty K_2) < m$ if $3n < x_m$.
\end{theorem}

\begin{theorem}
\label{thm:cartcycgen}
     Let $n \geq 3$ be an integer. Then $\psi_o(C_n \boxempty C_\ell) < m$ if $2n\ell < x_m$.
\end{theorem}

The next result talks about the oriented achromatic number of regular graphs, and is an improvement of Theorem 3.14 in \cite{PD202448}.

\begin{theorem}
\label{thm:cartreg}
     Let $G$ be an $r$-regular graph of order $n$ with $r \geq 1$. Then $\psi_o(G) < m$ if  $nr < 2x_m$.
\end{theorem}

The proofs of these results work exactly the same as in \cite{PD202448}. While we use $x_m$ as the upper bound on the number of arcs in an absolute oriented clique, \cite{PD202448} uses another upper bound which they obtained using a different construction.

\section{Conclusions}
\label{sec:sec6}
\noindent (1) \textbf{Regarding the upper bounds of $f_2(n)$:}
To derive their upper bound, F\"{u}redi, Horak, Pareek and Zhu~\cite{furedi1998minimal} used a recursive construction of absolute oriented cliques quite similar to what we have used. The difference between their construction and ours is that for the construction of each absolute oriented clique $\overrightarrow{O}_n = \overrightarrow{O}_{n_1} \ltimes \overrightarrow{O}_{n_2}$, we ensure that $|A(\overrightarrow{O}_{n_1})| + |A(\overrightarrow{O}_{n_2})|$ is minimized. This ensures that the upper bound that we obtain is better than the one obtained in \cite{furedi1998minimal}.

It is quite easy to calculate the values of the sequence $((x_n))$ for any $n$ using a computer program. In the graph in Figure~\ref{gr:comp}, we plot the number of vertices $n$ on the $x$-axis versus the upper bound on the minimum number of arcs in an absolute oriented clique on $n$ vertices as given by Kostochka, Luczak, Simonyi and Sopena~\cite{kostochka1999minimum}(KLSS), F\"{u}redi, Horak, Pareek and Zhu~\cite{furedi1998minimal}(FHPZ) and our result ($x_n$). This graph, together with some specific values mentioned in Table~\ref{tab:comp}, give a more visual representation of how our upper bound is the best one yet.

 \begin{figure}
     \centering
     \begin{tikzpicture}
\begin{axis}[
    xmin = 0, xmax = 1000,
    ymin = 0, ymax = 10000,
    % xtick distance = 2.5,
    % ytick distance = 0.5,
    grid = both,
    minor tick num = 1,
    major grid style = {lightgray},
    minor grid style = {lightgray!25},
    width = \textwidth,
    height = 0.7\textwidth,
    legend cell align = {left},
    legend pos = north west]
    \addplot[
        domain = 0:1000,
        samples = 1000,
        smooth, 
        thick,
        blue,
    ] {ceil(log2(x))*(x - ceil(log2(x)))};
    \addplot[
        domain = 0:1000,
        samples = 1000,
        smooth,
        thick,
        red,
    ] {floor(x*log2(x) - 3*x/2)};
    \addplot[
        smooth,
        thin,
        black,
    ] file[skip first] {Book1.prn};
    \legend{\textcolor{blue}{KLSS}, \textcolor{red}{FHPZ}, $x_n$}
\end{axis}
\end{tikzpicture}
     \caption{Comparision of $f_2(n)$ upper bounds as obtained by the various results mentioned in this paper}
     \label{gr:comp}
 \end{figure}

\begin{table}[h!]
\centering
\begin{tabular}{||c | c | c| c||} 
 \hline
$n$ & KLSS & FHPZ & \textbf{$x_n$}\\ [0.5ex] 
\hline
$10$ & 24 & 18 & \textbf{18}\\
\hline
$10^2$ & 651 & 514 & \textbf{467}\\
\hline
$10^3$ & 9900 & 8465 & \textbf{7976}\\
\hline
$10^4$ & 139804 & 117877 & \textbf{112727}\\
\hline
$10^5$ & 1699711 & 1510964 & \textbf{1453411}\\
\hline
$10^6$ & 19999600 & 18431568 & \textbf{17927158}\\[0.5ex]
\hline
\end{tabular}
\caption{Comparision of known upper bounds of $f_2(n)$ for some $n$.}
\label{tab:comp}
\end{table}

\medskip

\noindent (2) \textbf{Exact value of $f_2(n)$:}
We wrote a program that checked the validity of Conjecture~\ref{conj:main} for small number of vertices. We were able to confirm the conjecture for $n \leq 12$. For $n = 12$, the program checked more than $1.6 \times 10^{11}$ oriented graphs. Perhaps unsurprisingly, for $n \geq 13$, the number of graphs to be checked increases drastically and it was no longer feasible to check it. It would be an interesting problem to identify structural properties of absolute oriented cliques that could help us to significantly reduce the number of cases that need to be checked by a program.

Based on our observations in this work, we make the following conjecture.

\begin{conjecture}
\label{conj:sub1}
    For all $n > 5$, there exists an oriented clique with $n$ vertices and $f_2(n)$ arcs such that it has a cut vertex.
\end{conjecture}

Observe that Conjecture~\ref{conj:sub1}, together with Theorem~\ref{thm:upper} implies Conjecture~\ref{conj:main}. We also conjecture a weaker version of Conjecture~\ref{conj:sub1} as follows.

\begin{conjecture}
\label{conj:sub2}
    For all $n > 5$, there exists an oriented clique with $n$ vertices and $f_2(n)$ arcs such that it has a vertex of degree $n - 1$.
\end{conjecture}

We believe that Conjecture~\ref{conj:sub2} might be easier to prove than Conjecture~\ref{conj:sub1}. It is interesting to see if proving Conjecture~\ref{conj:sub2} would imply Conjecture~\ref{conj:sub1}.

\medskip

\noindent (3) \textbf{Colored mixed graphs:}
A broad generalization of oriented graphs and their homomorphisms, introduced by Ne\v{s}et\v{r}il and Raspaud~\cite{nevsetvril2000colored} is the notion of colored homomorphisms of colored mixed graphs. 
To elaborate, colored mixed graphs, or $(n,m)$-graphs are graphs having $n$ different types of arcs and $m$ different types of edges. Notice that, a (simple) $(n,m)$-graph for $(n,m) = (1,0)$ is nothing but an oriented graph.  The concept analogous to absolute oriented clique is present in the literature~\cite{bensmail2017analogues} for $(n,m)$-graphs as well. Thus it is a natural future research direction to extend the notion of 
$f_2(n)$ in the setup of  $(n,m)$-graphs and study its behavior. Moreover, one may also consider studying the natural extension of the concept of oriented achromatic number for $(n,m)$-graphs.

% of The concept of a clique has different analogues in various types of graph structures with respect to the homomorphisms defined for them. For example, \textit{$(m, n)$-cliques} are defined in the context of colored mixed graphs~\cite{}, \textit{signed cliques} are defined in the context of signed graphs~\cite{naserasr2015homomorphisms}, etc. Just like for absolute oriented cliques, we can define functions which give us the sizes of the smallest cliques in these contexts. A possible direction for research would be to study functions of this type and try to determine some of their properties.

\bibliographystyle{abbrv}
\bibliography{references}

\begin{thebibliography}{10}

\bibitem{bensmail2017analogues}
J.~Bensmail, C.~Duffy, and S.~Sen.
\newblock Analogues of cliques for (m, n)-colored mixed graphs.
\newblock {\em Graphs and Combinatorics}, 33:735--750, 2017.

\bibitem{BORODIN2001150}
O.~Borodin, D.~Fon-Der-Flaass, A.~Kostochka, A.~Raspaud, and E.~Sopena.
\newblock On deeply critical oriented graphs.
\newblock {\em Journal of Combinatorial Theory, Series B}, 81(1):150--155,
  2001.

\bibitem{courcelle1994monadic}
B.~Courcelle.
\newblock The monadic second order logic of graphs {VI}: On several
  representations of graphs by relational structures.
\newblock {\em Discrete Applied Mathematics}, 54(2):117--149, 1994.

\bibitem{das2015oriented}
S.~Das, S.~Mj, and S.~Sen.
\newblock On oriented relative clique number.
\newblock {\em Electronic Notes in Discrete Mathematics}, 50:95--101, 2015.

\bibitem{das2018study}
S.~Das, S.~Prabhu, and S.~Sen.
\newblock A study on oriented relative clique number.
\newblock {\em Discrete Mathematics}, 341(7):2049--2057, 2018.

\bibitem{dawes1987arc}
R.~Dawes and H.~Meijer.
\newblock Arc-minimal digraphs of specified diameter.
\newblock {\em J. Combin. Math. and Combin. Comput}, 1:85--96, 1987.

\bibitem{duffy2019oriented}
C.~Duffy, G.~MacGillivray, and {\'E}.~Sopena.
\newblock Oriented colourings of graphs with maximum degree three and four.
\newblock {\em Discrete Mathematics}, 342(4):959--974, 2019.

\bibitem{duffy2023}
C.~Duffy, P.~D. Pavan, R.~B. Sandeep, and S.~Sen.
\newblock On deeply critical oriented cliques.
\newblock {\em Journal of Graph Theory}, 104(1):150--159, 2023.

\bibitem{dybizbanski2020oriented}
J.~Dybizba{\'n}ski, P.~Ochem, A.~Pinlou, and A.~Szepietowski.
\newblock Oriented cliques and colorings of graphs with low maximum degree.
\newblock {\em Discrete Mathematics}, 343(5):111829, 2020.

\bibitem{erdos1966problem}
P.~Erdos, A.~R{\'e}nyi, and V.~S{\'o}s.
\newblock On a problem of graph theory.
\newblock {\em Studica Sci. Math Hungr. 1}, 1:215--235, 1966.

\bibitem{furedi1998minimal}
Z.~F{\"u}redi, P.~Horak, C.~M. Pareek, and X.~Zhu.
\newblock Minimal oriented graphs of diameter 2.
\newblock {\em Graphs and Combinatorics}, 14(4):345--350, 1998.

\bibitem{hell2004graphs}
P.~Hell and J.~Ne{\v{s}}et{\v{r}}il.
\newblock {\em Graphs and homomorphisms}, volume~28.
\newblock OUP Oxford, 2004.

\bibitem{katona1967problem}
G.~Katona and E.~Szemer{\'e}di.
\newblock On a problem of graph theory.
\newblock {\em Studia Scientiarum Mathematicarum Hungarica}, 2:23--28, 1967.

\bibitem{kirgizov2016complexity}
S.~Kirgizov, R.~Duvignau, and J.~Bensmail.
\newblock The complexity of deciding whether a graph admits an orientation with
  fixed weak diameter.
\newblock {\em Discrete Mathematics \& Theoretical Computer Science}, 17, 2016.

\bibitem{klostermeyer2004analogues}
W.~F. Klostermeyer and G.~MacGillivray.
\newblock Analogues of cliques for oriented coloring.
\newblock {\em Discussiones Mathematicae Graph Theory}, 24(3):373--387, 2004.

\bibitem{kostochka1999minimum}
A.~V. Kostochka, T.~Luczak, G.~Simonyi, and E.~Sopena.
\newblock On the minimum number of edges giving maximum oriented chromatic
  number.
\newblock {\em DIMACS Series in Discrete Mathematics and Theoretical Computer
  Science}, 49:179--182, 1999.

\bibitem{nandy2016outerplanar}
A.~Nandy, S.~Sen, and {\'E}.~Sopena.
\newblock Outerplanar and planar oriented cliques.
\newblock {\em Journal of Graph Theory}, 82(2):165--193, 2016.

\bibitem{nevsetvril2000colored}
J.~Ne{\v{s}}et{\v{r}}il and A.~Raspaud.
\newblock Colored homomorphisms of colored mixed graphs.
\newblock {\em Journal of Combinatorial Theory, Series B}, 80(1):147--155,
  2000.

\bibitem{PD202448}
P.~D. Pavan and Éric Sopena.
\newblock On the oriented achromatic number of graphs.
\newblock {\em Discrete Applied Mathematics}, 347:48--61, 2024.

\bibitem{sopena1997chromatic}
E.~Sopena.
\newblock The chromatic number of oriented graphs.
\newblock {\em Journal of Graph Theory}, 25(3):191--205, 1997.

\bibitem{sopena2014complete}
{\'E}.~Sopena.
\newblock Complete oriented colourings and the oriented achromatic number.
\newblock {\em Discrete Applied Mathematics}, 173:102--112, 2014.

\bibitem{znam1970minimal}
S.~Znam.
\newblock The minimal number of edges of a directed graph with given diameter.
\newblock {\em Acta Fac. Rerum Natur. Univ. Comenian. Math. Publ}, 24:181--185,
  1970.

\end{thebibliography}

\end{document}